\documentclass[onecolumn,superscriptaddress,longbibliography]{revtex4-2}
\usepackage[ruled,vlined]{algorithm2e}

\usepackage{complexity}
\usepackage[margin=1in]{geometry}
\usepackage[utf8]{inputenc}
\usepackage{bbold}
\usepackage[normalem]{ulem}
\usepackage[dvipsnames]{xcolor}

\usepackage{graphicx}
\usepackage{multirow}
\usepackage{amsmath,amsthm}
\usepackage{mathtools}
\usepackage{amssymb,amsfonts}
\usepackage{soul}
\usepackage{verbatim}
\usepackage{bm}
\usepackage[colorlinks=true, linkcolor=blue, citecolor=blue, urlcolor=blue]{hyperref}
\usepackage{tikz}
\usepackage{xspace}
\usepackage{import}
\usepackage[caption=false]{subfig}
\usepackage{bbold}
\usepackage{qcircuit}
\usepackage{chemformula}

\usepackage{dsfont}

\usepackage{adjustbox}

\newtheorem{thm}{Theorem}
\newtheorem{lemma}[thm]{Lemma}
\newtheorem{corollary}[thm]{Corollary}
\newtheorem{defn}[thm]{Definition}

\newtheorem{prop}[thm]{Proposition}

\newcommand{\ket}[1]{|#1\rangle}
\newcommand{\bra}[1]{\langle#1|}
\newcommand{\braket}[2]{\langle #1|#2\rangle}
\newcommand{\ketbra}[2]{| #1\rangle\!\langle #2 |}

\newcommand\norm[1]{\left\lVert#1\right\rVert}
\newcommand{\lb}{\left(}
\newcommand{\rb}{\right)}

\newcommand{\UofT}{\affiliation{
Department of Computer Science, University of Toronto, Canada}}

\newcommand{\UofTP}{\affiliation{
Department of Physics, University of Toronto, Canada}}

\newcommand{\PNNL}{\affiliation{Pacific Northwest National Laboratory, Richland WA, USA}}

\newcommand{\CIFAR}{\affiliation{Canadian Institute for Advanced Studies, Toronto, Canada}}

\newcommand{\BIQuantum}{\affiliation{
Quantum Lab, Boehringer Ingelheim, 55218 Ingelheim am Rhein, Germany}}
\begin{document}

\title{Dividing and Conquering the Van Vleck Catastrophe}

\author{Sophia Simon}
\email{sophia.simon@mail.utoronto.ca}
\UofTP

\author{Gian-Luca R. Anselmetti}
\BIQuantum

\author{Raffaele Santagati}
\email{raffaele.santagati@boehringer-ingelheim.com}
\BIQuantum

\author{Matthias Degroote}
\BIQuantum

\author{Nikolaj Moll}
\BIQuantum

\author{Michael Streif}
\BIQuantum

\author{Nathan Wiebe}
\email{nathan.wiebe@utoronto.ca}
\UofT
\PNNL
\CIFAR

\begin{abstract}
    The quantum-computational cost of determining ground state energies through quantum phase estimation depends on the overlap between an easily preparable initial state and the targeted ground state. The Van Vleck orthogonality catastrophe has frequently been invoked to suggest that quantum computers may not be able to efficiently prepare ground states of large systems because the overlap with the initial state tends to decrease exponentially with the system size, even for non-interacting systems. We show that this intuition is not necessarily true. Specifically, we introduce a divide-and-conquer strategy that repeatedly uses phase estimation to merge ground states of increasingly larger subsystems. We provide rigorous bounds for this approach and show that if the minimum success probability of each merge is lower bounded by a constant, then the query complexity of preparing the ground state of $N$ interacting systems is in $O(N^{\log\log(N)} {\rm poly}(N))$, which is quasi-polynomial in $N$, in contrast to the exponential scaling anticipated by the Van Vleck catastrophe. We also discuss sufficient conditions on the Hamiltonian that ensure a quasi-polynomial running time.
\end{abstract}

\maketitle


\section{Introduction}

The simulation of quantum systems is one of the major drivers behind the development of quantum computers. A fundamental application of such quantum simulations is the preparation of a Hamiltonian's ground state and the estimation of its energy~\cite{Cao2019Rev,McArdle2020,Aspuru-Guzik2005,Bauer2020}. This task is not only crucial for studying complex systems in, e.g., physics and chemistry, but also finds critical applications in industries, for instance, to understand the process of nitrogen fixation~\cite{reiher2017elucidating}, study the metabolism of drugs~\cite{goings2022reliably, caesura2025faster, Santagati2024}, or to develop new battery cathodes~\cite{rubin2023fault, kim2022fault}.

The leading algorithm for determining ground state energies of Hamiltonians on future fault-tolerant quantum computers is Quantum Phase Estimation (QPE). Given a Hamiltonian $H$ and some initial state $\ket{\psi_i}$, QPE allows us to estimate the eigenvalues of $H$ by using the phase information collected from simulating under the action of the Hamiltonian. As a projector method, the probability of measuring one of the eigenvalues is directly linked to the amplitude of the corresponding eigenstate present in the initial state. 
Thus, the computational cost of estimating the ground state energy via QPE is determined by two factors: (i) the cost per measurement, which is dominated be the number of quantum gates required to simulate time evolution under the Hamiltonian, and (ii) the number of measurements, which depends on the overlap between the initial state and the ground state whose energy is being estimated. In the most basic approach, the ground state can be identified by statistically sampling the lowest energy states, provided there is a non-negligible eigenvalue gap $E_1 - E_0 \geq \gamma$ between the ground state energy $E_0$ and the first excited energy $E_1$, and the observed energy falls below this gap. Although more advanced techniques for ground state preparation have emerged~\cite{Ge2019, Lin2020}, all these methods rely on the fact that an initial state shares a non-negligible overlap with the Hamiltonian’s ground state. Some recent proposals utilize ideas from dissipative dynamics and thermal state preparation to prepare ground states without requiring a good initial state~\cite{cubitt2023dissipative, motlagh2024ground, Ding2024_Lindbladian, hagan2025thermodynamic, zhan2025rapidquantumgroundstate}. However, the time complexity of these algorithms tends to depend on other quantities, such as the mixing time of a carefully chosen Lindbladian, which are generally also difficult to bound. Understanding the trade-offs between the different methods for practical applications remains an important open problem.

First resource estimates for determining the ground state energy of various industrially relevant systems \cite{lee2021even, goings2022reliably, von2021quantum, caesura2025faster, low2025fast} with quantum phase estimation (QPE) assumed the availability of an initial state that has a substantial overlap with the ground state. Nonetheless, recent studies indicate that for complex systems, such as the FeMoco, the overlap between commonly used initial states and the desired ground state can vanish exponentially with the system size~\cite{Lee2023}, leading to an exponential increase in the computational time of QPE. In response to this challenge, researchers have developed methods to enhance the initial state overlap, such as leveraging the gauge freedom of basis rotations~\cite{ollitrault2024enhancing} and constructing better initial states via sums of Slater determinants~\cite{Tubman2018} or through tensor networks~\cite{fomichev2024initial, berry2024rapid}.

A major concern that can be levied against these results stems from the question of whether an accurate approximation to the ground state can be prepared in polynomial time on a quantum computer.  Theoretical results show that, in general, the answer to the question is negative as the preparation of the ground state in a fixed basis is~\QMA-hard~\cite{o2022intractability}.  However, it is unclear whether these hard instances correspond to physically realistic molecules as even a quantum computer is unlikely to be able to prepare such states in polynomial time. 
On the other hand, even for physically relevant systems it can be nontrivial to prepare their ground states on a quantum computer. Van Vleck argued that if we focus on the thermodynamic limit of quantum mechanics then any state that has constant overlap with the ground state of a single copy, will have exponentially shrinking overlap with the collective ground state as the number of copies of these states increases~\cite{Vanvleck1936}. 
As the cost of quantum approaches to the electronic structure problem depend inversely on this overlap~\cite{reiher2017elucidating,lee2021even, goings2022reliably, von2021quantum, caesura2025faster, low2025fast}, such an orthogonality catastrophe has been suspected to be a serious obstacle to the development of scalable quantum algorithms for chemistry~\cite{Lee2023}.

In this paper, we address the question of whether it is possible to avoid the Van Vleck catastrophe under specific circumstances. To that end, we propose a divide-and-conquer strategy for constructing an initial state that effectively counters the exponentially decaying overlap and ensures that the overall gate complexity scales only quasi-polynomially with the system size.  
The main idea behind our algorithm is to divide the full system into smaller subsystems and then use phase estimation to project onto the individual ground states via a measure-and-repeat-until-success strategy before merging pairs into larger systems. 
This iterative process is continued until the whole system is rebuilt. Using this procedure, we demonstrate, both analytically and numerically, that it is possible to circumvent the Van Vleck catastrophe in specific instances of weakly interacting subsystems.

The remainder of the paper is structured as follows. We discuss the Van Vleck catastrophe and our divide and conquer approach to state preparation in Section~\ref{sec:divConq}. Then, in Section~\ref{sec:perturbed}, we review results from eigenvalue and eigenvector perturbation theory which allow us to understand how the success probability scales with the strength of the interaction and the gap in the eigenvalues and use these results to provide sufficient conditions for sub-exponential scaling for ground state preparation. Next, we discuss the impact of entanglement area laws and volume laws on the ground state and show that a volume law of entanglement immediately implies that the system is too strongly correlated for the divide and conquer approach to succeed with high-probability using sub-exponential resources. We then provide numerical experiments in Section~\ref{sec:numerics} which show that our sufficient conditions for sub-exponential scaling may actually be rather pessimistic for one-dimensional systems and that polynomial scaling may be more widely achievable than our analytical results suggest.  In Section~\ref{sec:fermion} we clarify the issues of separability that fermionic statistics may seem to create for our technique and show that the non-locality that appears in anti-symmetrizing states or operators does not impede the divide and conquer approach.  Finally, we conclude in Section~\ref{sec:conclusion} where we also discuss future directions.

\section{Combating the Van Vleck Catastrophe}
\label{sec:divConq}

Before beginning with our discussion of how to combat the Van Vleck catastrophe, we will formally discuss the origins of it and provide a short proof of why it emerges.  

\begin{lemma}[Van Vleck Catastrophe, adapted from~\cite{Vanvleck1936, Kohn1999}]
\label{lem:vanvleck}
    Let $\ket{\psi_{*}}$ and $\ket{\psi_{i}}$ be quantum states such that for some $\delta\in (0,1)$,
    $$
        |\braket{\psi_{*}}{\psi_{i}}|\le 1-\delta.
    $$
    We then have that for any integer $N$, the overlap between the tensor product of $N$ copies of $\ket{\psi_{i}}$ and $\ket{\psi_{*}}$ is upper bounded as follows:
    $$
        |\braket{\psi_{*}^{\otimes N}}{\psi_{i}^{\otimes N}}|\le (1-\delta)^{N}.
    $$
\end{lemma}

\begin{proof}
    The proof follows straightforwardly by induction.  The base case in the inductive proof is assumed to be true already.  This means that we only need to demonstrate the inductive step.  Assume that the hypothesis is true for $N=r$ for some $r\ge 1$. Then we have that
    \begin{equation}
        |\braket{\psi_{*}^{\otimes r+1}}{\psi_{i}^{\otimes r+1}}| = |(\bra{\psi_{*}^{\otimes r}}\otimes \bra{\psi_{*}}) (\ket{\psi_{i}^{\otimes r}}\otimes \ket{\psi_{i}}) |= |\braket{\psi_{*}^{\otimes r}}{\psi_{i}^{\otimes r}}||\braket{\psi_{*}}{\psi_{i}}|\le (1-\delta)^{r+1},
    \end{equation}
    as desired. Thus, if $\delta > 0$ then the overlap shrinks exponentially with $N$. 
\end{proof}

As shown in Lemma~\ref{lem:vanvleck}, taking tensor products of quantum states with non-unit overlap inevitably leads to an exponential decrease in the overall overlap. In general, even an initial state with good overlap for the individual subsystems will fail to produce the overall ground state with high probability when increasing the size of the system. This is true even in the case where there are no interactions between the subsystems. This in particular should raise a question in the reader's mind: if this exponential decay of overlap holds even for non-interacting states, then could we avoid this problem using a more clever method of preparing a state?  In fact, the answer is yes.

To see why this is true, let us consider a related problem of flipping coins.  Let us assume that we wish to prepare $100$ coins all in the state ``heads''.  One algorithm to do this would be to repeatedly flip all of the coins at once until we succeed.  In that case, in complete accordance with Lemma~\ref{lem:vanvleck}, the expected number of coin flips  needed for this process is $100 \times 2^{100}$.  However, if we flip the first coin until we get a ``heads'' outcome, then proceed to the second coin and so on, then the expected number of coin flips needed to prepare the state is only $200$.  This intuition holds true in the non-interacting case as well: we can divide our state preparation problem up into a large number of uncorrelated state preparations and repeat each until we get a success.  The central point of our paper is to take this insight from the non-interacting case to show that in general we can prepare weakly interacting ground states in sub-exponential time.

Our central tactic for defeating the Van Vleck catastrophe is to use a divide-and-conquer strategy to break the problem of preparing the ground state up into the problem of merging together groups of approximate ground states of smaller subsystems.  The idea is to use phase estimation to fuse together these smaller systems recursively, and when a failure to project is detected, we discard the quantum state that we have prepared and generate a new one.  In this manner, we break apart a problem of merging together $2^p$ interacting systems to the problem of merging $p$ chains of approximate ground states.  We then show that, for appropriately weakly interacting systems with significant eigenvalue gaps, the ground state can be prepared in polynomial time despite the appearance of the Van Vleck catastrophe. We will then show numerically that while sufficient, this criterion is not necessary. In particular, we provide an example where the success probability is very large despite the spectral gaps not being large compared to the norm of the interaction Hamiltonians.

\subsection{Dividing and Conquering Van Vleck}

We consider a system composed of $N=2^p$ basic subsystems, where interactions between multiple subsystems follow a tree-like structure. For simplicity, we will focus on perfect binary trees, but our results can easily be generalized to arbitrary tree structures. Below, we define our input model in more detail.

\begin{defn}[Input Model]
    Let $N = 2^p$ for some integer $p$ be the number of basic subsystems. Each subsystem may contain a different number of particles (or qubits). The Hamiltonian $H$ of the entire system can be represented as a perfect binary tree of height $p$ such that
    \begin{equation}
        H = \sum_{n=0}^{p} \sum_{s \in \{ 0, 1 \}^n} A_{s},
    \end{equation}
    where $A_s$, with $s \in \{ 0, 1\}^n$ being a binary string of length $ 0 \leq n \leq p$, corresponds to a node at level $n$ of the binary tree. More specifically, if $n = p$, then $A_s$ corresponds to a leaf node of the binary tree and acts nontrivially only on the basic subsystem with label $s$. Else, define $S_s$ to be the set of all $p$-bit subsystem labels whose $n$ most significant bits are equal to $s$. Then $A_s$ encodes all the interactions between the subsystems in $S_s$ which are not already included in its child nodes, $A_{s0}$ and $A_{s1}$, see Fig.~\ref{fig:binary_tree} for a visualization. We denote the term labeled by the empty string, corresponding to $n=0$, by $A_{*}$. Further, we define $H_s$ to be the Hamiltonian associated with the set $S_s$, i.e., $ H_s$ is made up of all terms acting nontrivially solely on subsystems in $S_s$. The full Hamiltonian $H$ corresponds to $s$ being the empty string. Lastly, for each $H_s$, we also define the following unitary operator:
    \begin{equation}
        U_s := e^{ i \pi H_s/2\norm{H_s}}.
    \end{equation}
\label{def:input_ham}
\end{defn}

An example decomposition of a Hamiltonian for a system composed of $N = 8$ subsystems according to Definition~\ref{def:input_ham} is provided in Fig.~\ref{fig:binary_tree}.

\begin{figure}
    \centering
    \includegraphics[width=0.99\linewidth]{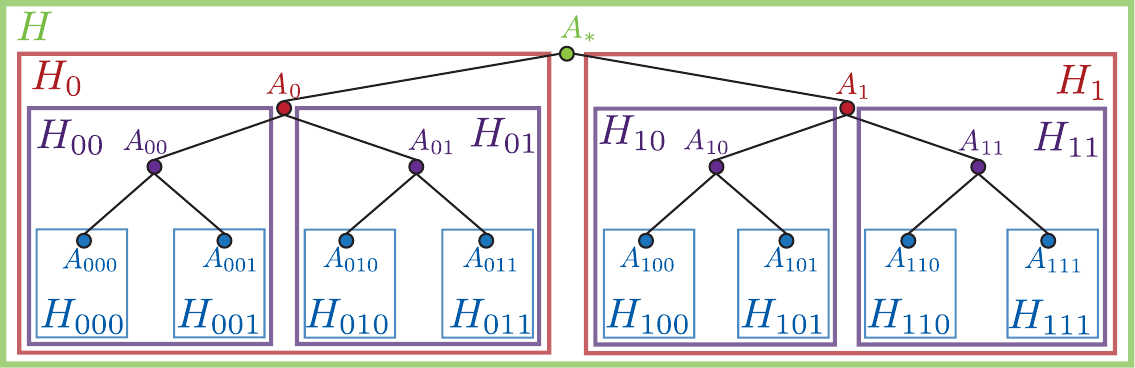}
    \caption{Binary tree decomposition of a Hamiltonian $H$ for a system composed of $8$ basic subsystems. The leaf nodes correspond to the Hamiltonian terms that act solely on the basis subsystems. Multi-system interactions are introduced in the higher levels of the tree. The various subsystem Hamiltonians and the terms that they include are indicated by the colored boxes.}
    \label{fig:binary_tree}
\end{figure}

Let us now discuss how to prepare the ground state of a composite system by recursively preparing the ground states of smaller subsystems and then combining them via phase estimation. The key idea that allows us to avoid an exponentially decreasing success probability, and hence an exponential cost, is to use a measure-until-success strategy at each stage of the recursion. In case of a failure at any stage, we only need to redo the preparation on the affected subsystems.
Now, consider a Hamiltonian $H$ on $N = 2^p$ subsystems according to Definition~\ref{def:input_ham} and let $U_s$ be the unitary time evolution operator associated with the subsystem Hamiltonian $H_s$. Further, let $V_0, V_1, \dots, V_{N-1}$ be oracles that prepare the ground state of subsystem $0, 1, \dots, N-1$, respectively. The aim is to upper bound the overall number of queries to $V_0, V_1, \dots V_{N-1}$ and all $U_s$ required to prepare the ground state of the entire composite system. The following proposition is a first step towards this goal as it will allow us to bound the propagation of errors in the state preparation.

\begin{prop}
    Let $\delta \in [0,1]$ and let $n \geq 1$ be an integer. Then
    \begin{equation}
        \lb 1- \frac{\delta}{n} \rb^n \geq 1 - \delta.
    \end{equation}
\label{prop:probability}
\end{prop}

\begin{proof}
    We prove the above proposition via induction. First, we check that the base case with $n=1$ is true:
    \begin{equation}
        \lb 1- \frac{\delta}{1} \rb^1 \geq 1 - \delta
    \end{equation}
    as desired. For the induction step, we assume that $\lb 1- \frac{\delta}{n} \rb^n \geq 1 - \delta$ is true and aim to prove that it must then also hold for $n+1$, i.e.~we need to show that $\lb 1- \frac{\delta}{n+1} \rb^{n+1} \geq 1 - \delta$. Let $\delta' := \frac{n}{n+1} \delta$. Then
    \begin{equation}
    \begin{split}
         \lb 1- \frac{\delta}{n+1} \rb^{n+1} &= \lb 1- \frac{\delta}{n+1} \rb \lb 1- \frac{\delta}{n+1} \rb^{n} = \lb 1- \frac{\delta}{n+1} \rb \lb 1- \frac{\delta'}{n} \rb^{n} \\
         &\geq  \lb 1- \frac{\delta}{n+1} \rb \lb 1- \delta' \rb \\
         &\geq 1 - \frac{\delta}{n+1} - \frac{n \delta}{n+1} \\
         &\geq 1 - \delta.
    \end{split}
    \end{equation}
\end{proof}

Now we are ready to state the main theorem which provides an upper bound on the complexity of preparing the ground state of a set of interacting systems via a divide-and-conquer strategy. We will see that the gate complexity need not scale exponentially with the number of systems in contrast to what a naive application of the Van Vleck catastrophe argument would suggest. In the following, unless stated otherwise, we will use $\norm{M}$ to refer to the induced 2-norm (i.e.~spectral norm) of a matrix $M$.

\begin{thm}[Divide-and-Conquer Ground State Preparation]
    Let $H$ be a Hamiltonian on $N = 2^p$ subsystems according to Definition~\ref{def:input_ham} and let $\ket{\psi_s}$ be the ground state of $H_s$ for all $s \in \left\{ \{ 0, 1\}^j | j \in \mathbb{Z}_p \right\}$ subject to the promise that  $\forall  s$ $\|H_s\| \le H_{\max}$. Further, let $\gamma_s$ be the spectral gap of $H_s$ with $\gamma_s \geq \gamma_{\min}$ for all $s$ and let $U_s = \exp \lb i \pi  H_s/2\norm{H_s} \rb$ as in Definition~\ref{def:input_ham}.  
    Assume that we have access to a set of oracles $V = \left[V_0, V_1, \dots, V_{N-1}\right]$, which prepare the ground state of subsystem $0, 1, \dots, N-1$, respectively. 
    Then there exists a quantum algorithm that can prepare the ground state $\ket{\psi_*}$ of the entire system with success probability at least $1 - \delta$ using
    a total number of queries to all the $V$ oracles that obeys
    \begin{equation}
        N_V \in O \lb N^{1 + \log (\pi^2/4r^2) + \log \log \lb 1/\delta \rb + \log \log \lb N^2 \rb} \rb,
    \end{equation}
    and a total number of applications of all $U_s$ that scale as
    \begin{equation}
        N_{U_s} \in O \lb \frac{H_{\max}}{\gamma_{\min}} N^{1 + \log (\pi^2/4r^2) + \log \log \lb 1/\delta \rb + \log \log \lb N^2 \rb} \rb,
    \end{equation}
    where $r \in [0,1]$ is a lower bound on the overlap between $\ket{\psi_{s0}}\ket{\psi_{s1}}$ and $\ket{\psi_{s}}$ $\forall s$. 
\label{thm:divide+conquer}
\end{thm}

\begin{proof}
    Our general strategy is to prepare the overall ground state via the binary tree structure apparent in Hamiltonians according to Definition~\ref{def:input_ham}. Each leaf node represents one of the $N = 2^p$ subsystems. We will refer to the leaf nodes as layer $p$. The next layer, which we will call layer $p-1$, contains $2^{p-1} = N/2$ nodes. Each of those represents a composite system made up of two individual subsystems. The last layer, layer $0$, has only a single node, the root node, and represents the entire system. Fig.~\ref{fig:binary_tree} visualizes the binary tree structure for $p=3$.

    We prove the theorem inductively by showing that the ground state $\ket{\psi_{s}}$ of a collection of subsystems labeled by a bit string $s$ of length $j$  with $j \in \left\{ 0, 1, \dots, p \right\}$ (e.g. $s \in \{ 00, 01,10,11 \}$ for $j=2$) can be prepared with probability at least $1-\delta$ using
    \begin{equation}
        N_V(j, \delta) \le C_1 2^{(p-j)\left( 1 
        + \log (\pi^2/4r^2) + \log \log \lb 1/\delta \rb + \log (2(p-j)) \right)} 
    \end{equation}
    queries to all oracles $V_0, \dots, V_{N-1}$ and
    \begin{equation}
        N_{U_s}(j, \delta) \le \frac{C_2 H_{\max}}{\gamma_{\min}} 2^{(p-j) \left( 1 + \log (\pi^2/4r^2) + \log \log \lb 1/\delta \rb + \log (2(p-j))\right)}
    \end{equation}
    queries to all unitaries $U_s$ with $s \in \left\{ \{ 0, 1\}^{m} | m \in \left\{ j, j+1, \dots, p \right\} \right\}$. Here, $C_1 \geq 1$ and $C_2 \geq 0$ are universal constants. The above inequalities form our induction hypothesis. Note that a string of length $j$ corresponds to a node in layer $j$ of the binary tree structure representing a collection of $2^{p-j}$ subsets.
    Let us first discuss the base case $j=p$, i.e.~consider one of the leaf nodes. In that case, we only need $1 \leq C_1$ query to one of the $V$  oracles in order to prepare the desired ground state. As no phase estimation is required, the number of queries to any $U_s$ is 0. Hence, the induction hypothesis holds. 
    
    Next, let us take a look at the induction step. We assume that the induction hypothesis holds for $j = n$ with $n \in \{1, 2, \dots, p \}$ and aim to prove that it must then also hold for $j = n - 1$.
    Let $\ket{\psi_{s'0}}$ and $\ket{\psi_{s'1}}$ be the ground states associated with the child nodes of the node labeled by the bit string $s'$ of length $n-1$ and let $\ket{\psi_{s'}}$ be the ground state associated with the node labeled $s'$. By assumption,
    \begin{equation}
        \left| \bra{\psi_{s'}}(\ket{\psi_{s'0}}\ket{\psi_{s'1}}) \right| \geq r.
    \end{equation}
    The idea now is to use phase estimation to prepare $\ket{\psi_{s'}}$ from the initial state $\ket{\psi_{s'0}}\ket{\psi_{s'1}}$. Specifically, recall that $U_{s'} = e^{i \pi H_{s'}/2\norm{H_{s'}}}$ and $\gamma_{s'}$ is the spectral gap of $H_{s'}$. In order to distinguish the ground state from the first excited state, we need to perform phase estimation within error $\frac{\gamma_{s'}}{2\norm{H_{s'}}}$ which can be achieved using $O \lb \log \lb \norm{H_{s'}}/\gamma_{{s'}} \rb \rb \subseteq O \lb \log \lb H_{\max}/\gamma_{\min} \rb \rb$ ancilla qubits and $C_2 \norm{H_{s'}}/\gamma_{s'} \leq C_2 H_{\max}/\gamma_{\min}$ applications of $U_{s'}$ for some constant $C_2$. Then phase estimation returns $\ket{\psi_{s'}}$ with success probability $\xi \geq 4r^2/\pi^2$~\cite{Cleve1998revisited}. Now, suppose we run the phase estimation circuit $k$ times. 
    To ensure that the probability of never observing a single success is upper bounded by $\delta'$, i.e.
    \begin{equation}
        P_{\mathrm{fail}}(k) = \lb 1 - \xi \rb^k \leq \delta',
    \end{equation}
    it suffices to choose $k = \frac{\log (1/\delta')}{\xi}$.
    Thus, the overall cost of preparing $\ket{\psi_{s'}}$ with success probability at least $1 - \delta'$ conditioned upon having prepared both $\ket{\psi_{s'0}}$ and $\ket{\psi_{s'1}}$ successfully is upper bounded by
    \begin{align}
        &k \quad \text{queries to the state preparation routines of $\ket{\psi_{s'0}}$ and $\ket{\psi_{s'1}}$,} \\
        &k \, C_2 \,H_{\max}/\gamma_{\min} \quad \text{applications of $U_{s'}$}.
    \end{align}
    According to the induction hypothesis, we can prepare both $\ket{\psi_{s'0}}$ and $\ket{\psi_{s'1}}$ each with probability at least $1 - \delta'$ using at most
    \begin{equation}
        2 \,N_{V}(n, \delta') \le 2 \, C_1 2^{(p-n)\left( 1 
        + \log (\pi^2/4r^2) + \log \log \lb 1/\delta' \rb + \log (2(p-n)) \right)} 
    \end{equation} 
    queries to all of the $V$ oracles and
    \begin{equation}
        2 \, N_{U_s}(n, \delta') \le 2 \, \frac{C_2 H_{\max}}{\gamma_{\min}} 2^{(p-n) \left( 1 + \log (\pi^2/4r^2) + \log \log \lb 1/\delta' \rb + \log (2(p-n))\right)}
    \end{equation}
    applications of all $U_s$ with $s \in \left\{ \{ 0, 1\}^{m} | m \in \left\{ n, n + 1, \dots, p \right\} \right\}$.
    The overall success probability of preparing $\ket{\psi_{s'}}$ is then lower bounded by $(1 - \delta ')^3$. By Proposition~\ref{prop:probability} it suffices to pick $\delta' = \frac{\delta}{3}$ to ensure that the overall success probability of preparing $\ket{\psi_{s'}}$ is at least $1- \delta$.
    The total number of queries to all $V$ oracles is then upper bounded as follows:
    \begin{equation}
    \begin{split}
        N_V(n-1, \delta) &\leq 2k N_v(n, \delta/3) \\
        &\leq 2k\, C_1 \, 2^{(p-n)\left( 1 
        + \log (\pi^2/4r^2) + \log \log \lb 1/\delta' \rb + \log (2(p-n)) \right)} \\
        &\leq 2 \frac{\log (3/\delta)}{\xi} C_1 2^{(p-n) \lb 1 
        + \log (\pi^2/4r^2) + \log \log \lb 3/\delta \rb + \log \log \lb 2^{2(p-n)} \rb \rb} \\
        &\leq C_1 2^{(p-n+1) \lb 1 + \log (\pi^2/4r^2) \rb}2^{\log \log (3/\delta) + (p-n) \log \log \lb 2^{2(p-n)} \times 3/\delta \rb} \\
        &\leq C_1 2^{(p-n+1) \lb 1 + \log (\pi^2/4r^2) \rb} 2^{\log \log \lb 2^{2(p - n + 1)}/\delta \rb + (p-n) \log \log \lb 2^{2(p - n + 1)}/\delta \rb} \\
        &\leq C_1 2^{(p-n+1) \lb 1 + \log (\pi^2/4r^2)  + \log \log (1/\delta) + \log \log \lb 2^{2(p-n+1)} \rb \rb} \\
        &\leq C_1 2^{(p-n+1) \lb 1 + \log (\pi^2/4r^2)  + \log \log (1/\delta) + \log \lb 2(p-n+1) \rb \rb}.
    \end{split}
    \end{equation}
    By a similar argument, we have that the total number of applications of all $U_s$, now with \\
    $s \in \left\{ \{ 0, 1\}^{m} | m \in \left\{n-1, n, \dots, p \right\} \right\}$, obeys
    \begin{equation}
    \begin{split}
        N_{U_s}(n-1, \delta) &\leq 2k N_{U_s}(n, \delta/3) \\
        &\leq 2k \, \frac{C_2 H_{\max}}{\gamma_{\min}} \, 2^{(p-n)\left( 1 
        + \log (\pi^2/4r^2) + \log \log \lb 1/\delta' \rb + \log (2(p-n)) \right)} \\
        &\leq \frac{C_2 H_{\max}}{\gamma_{\min}} 2^{(p-n+1) \lb 1 + \log (\pi^2/4r^2)  + \log \log (1/\delta) + \log \lb 2(p-n+1) \rb \rb}.
    \end{split}
    \end{equation}
    This demonstrates the induction step and so we have that the result follows from the base case. Specifically,
    \begin{equation}
    \begin{split}
        N_V = N_V(0, \delta) &\leq C_1 2^{p \lb 1 + \log (\pi^2/4r^2)  + \log \log (1/\delta) + \log \lb 2p \rb \rb} \\
        &\leq C_1 N^{1 + \log (\pi^2/4r^2)  + \log \log (1/\delta) + \log \log \lb N^2 \rb},
    \end{split}
    \end{equation}
    and similarly,
    \begin{equation}
    \begin{split}
        N_{U_s} = N_{U_s}(0, \delta) &\leq \frac{C_2 H_{\max}}{\gamma_{\min}} 2^{p \lb 1 + \log (\pi^2/4r^2)  + \log \log (1/\delta) + \log \lb 2p \rb \rb} \\
        &\leq \frac{C_2 H_{\max}}{\gamma_{\min}} N^{1 + \log (\pi^2/4r^2)  + \log \log (1/\delta) + \log \log \lb N^2 \rb}.
    \end{split}
    \end{equation} 
\end{proof}

Note that the above asymptotic bounds depend strongly on the scaling behavior of $r$. If there are no interactions between different subsystems, then $r=1$.
In the corollary below we discuss sufficient conditions for achieving a gate complexity that is quasi-polynomial in $N$.

\begin{corollary}[No Van Vleck Catastrophe]
    Consider the same setting as in Theorem~\ref{thm:divide+conquer} and assume that $H_{\max}/\gamma_{\min} \in O \lb \mathrm{poly} (N) \rb$. Further, assume that each $V_0, V_{1} , \dots V_{N-1}$ and each $U_s$ for all $s \in \left\{ \{ 0, 1\}^j | j \in [p-1] \right\}$ can be implemented with gate complexity in $O \lb \mathrm{poly} (N) \rb$. 
    If we are promised that $1/r^2 \leq N$, then the gate complexity of preparing the ground state $\ket{\psi_*}$ of the entire system with constant success probability $\geq 2/3$ is in
    \begin{equation}
        O \lb N^{\log (N) + \log \log \lb N \rb} \mathrm{poly} (N) \rb.
    \label{no_van_vleck_weak}
    \end{equation}
    If we have the stronger promise that $1/r^2 \in O(1)$, then the number of required gates scales as
    \begin{equation}
        O \lb N^{\log \log \lb N \rb} \mathrm{poly} (N) \rb.
    \label{no_van_vleck_strong}
    \end{equation}
\label{cor:no_van_vleck}
\end{corollary}

\begin{proof}
    Follows directly from Theorem~\ref{thm:divide+conquer}.  Specifically, if we examine the $r$-dependent term in the expansion in the case where $1/r^2 \le N$, then
    \begin{equation}
        N^{\log(\pi^2/4r^2)} \le N^{\log(N \pi^2/4)} \in O \lb N^{\log(N)} {\rm poly}(N) \rb.
    \end{equation}
    Then, as $\delta \ge 2/3$, the overall gate complexity is given by
    \begin{equation}
        O\lb N^{\log(N) + \log \log(N^2)} {\rm poly}(N) \rb.
    \end{equation}
    We can repeat the same argument for $1/r^2 \in O(1)$.  In this case we obtain
    \begin{equation}
        N^{\log(\pi^2/4r^2)} \in O \lb {\rm poly}(N) \rb,
    \end{equation}
    which then gives us the required result in exactly the same manner as above.
\end{proof}

This shows that if the overlap between the ground state of a given node in the binary tree and the tensor product of the ground states associated with the corresponding child nodes is sufficiently large for every node in the binary tree, then the gate complexity for preparing the ground state of the entire system scales quasi-polynomially with $N$. In the best case scenario where $1/r^2 \in O(1)$ the gate complexity scales as $N^{\log \log \lb N \rb}$, up to polynomial factors. Note that $\log \log \lb N \rb \leq 4$ for $N \leq 10^{23}$. This means that for most practical intents and purposes, the gate complexity can be considered to be polynomial in $N$ if $r$ is lower bounded by a constant. In contrast, the Van Vleck catastrophe would predict a gate complexity which scales exponentially with $N$.

\section{Perturbation Theory} 
\label{sec:perturbed}

In the previous section, we showed that the Van Vleck catastrophe can be circumvented provided that the overlaps between ground states in adjacent layers of the binary tree is sufficiently large. In this section, we use perturbation theory to derive a lower bound $r$ on these overlaps.
To do so we will make use of the following theorem which provides an upper bound on the distance between the eigenvectors of two Hermitian matrices:

\begin{thm}[Simplified Davis-Kahan Theorem~\cite{bhatia2013matrix, DavisKahan2014}]
    Let $A, A' \in \mathbb{C}^{N \times N}$ be Hermitian with eigenvalues $\lambda_0 \leq \lambda_1 \leq  \dots \leq \lambda_{N-1}$ and $\lambda_0' \leq \lambda_1' \leq  \dots \leq \lambda_{N-1}'$, respectively. Furthermore, let $v_0, v_1, \dots, v_{N-1}$ and $v_0', v_1', \dots, v_{N-1}'$ be eigenvectors of $A$ and $A'$ such that $Av_j = \lambda_j v_j$ and $A'v_j' = \lambda_j' v_j'$. Then we have that $\forall j \in \{0,1,\dots, N-1\}$,
    \begin{equation}
        \sqrt{1-|\langle v_j, v_j'\rangle|^2} \leq \frac{\pi}{2} \frac{\norm{A-A'}}{\delta_j},
    \end{equation}
    where $\delta_j := \min \{ |\lambda_j - \lambda_{j-1}'|, |\lambda_j - \lambda_{j+1}'| \}$ with $\lambda_{-1}' := -\infty$ and $\lambda_N' := \infty$.
\label{thm:davis-kahan}
\end{thm}

Note that $\delta_j$ in the Davis-Kahan theorem describes the minimum distance between the unperturbed eigenvalue $\lambda_j$ and the perturbed spectrum $\lambda_0', \dots, \lambda_{N-1}'$ excluding $\lambda_j'$. We can lower bound $\delta_j$ in terms of the gap separating $\lambda_j$ from the unperturbed spectrum by using Weyl's inequality.

\begin{thm}[Weyl's Inequality \cite{Weyl1912perturbation}]
    Let $A, A' \in \mathbb{C}^{N \times N}$ be Hermitian with eigenvalues $\lambda_0 \leq \lambda_1 \leq  \dots \leq \lambda_{N-1}$ and $\lambda_0' \leq \lambda_1' \leq  \dots \leq \lambda_{N-1}'$, respectively. Then for any $j \in  \{0,1,\dots, N-1\}$,
    \begin{equation}
        | \lambda_j - \lambda_j'| \leq \norm{A - A'}.
    \end{equation}
\label{thm:weyl}
\end{thm}

Combining the Davis-Kahan theorem and Weyl's inequality thus allows us to derive a lower bound on the overlap between the ground states of an unperturbed Hamiltonian $H_0$ and a perturbed Hamiltonian $H_0 + H_{\mathrm{int}}$ solely in terms of the spectral gap of $H_0$ and the spectral norm of the perturbation, $\norm{H_{\mathrm{int}}}$. However, from an algorithmic perspective, we also need to take approximation errors into account. In general, we will not be able to implement the unitary $\exp \lb -i(H_0 + H_{\mathrm{int}})t \rb$ exactly. This means that quantum phase estimation, with the ground state of $H_0$ as input, will not exactly project onto the ground state of $H_0 + H_{\mathrm{int}}$ but rather some approximate ground state.
In order to quantify this approximation error, we introduce the notion of an effective Hamiltonian.

\begin{defn}[Effective Hamiltonian]
    Let $H \in \mathbb{C}^{N \times N}$ be a Hermitian matrix and let $U := \exp{(-iHt)}$ with $0 < t \leq \frac{1}{4\norm{H}}$ be the associated time evolution operator. Let $\widetilde{U}$ be a unitary approximation to $U$ such that $\norm{\widetilde{U} - U}  \leq \frac{1}{3}$.
    Then we call
    \begin{equation}
        \widetilde{H} := \frac{i}{t} \log \widetilde{U}
    \end{equation}
    the effective Hamiltonian of $\widetilde{U}$ where 
    \begin{equation}
        \log \widetilde{U} := \sum_{k=1}^\infty \frac{(-1)^{k+1}}{k} \lb \widetilde{U} - \mathbb{1} \rb^k.
    \end{equation}
\label{def:effective_ham}
\end{defn}

The choice of $t$ in the above definition ensures that the expression for $\log \widetilde{U}$ converges such that $\exp \lb \log \widetilde{U} \rb = \widetilde{U}$. We show this explicitly in the proof of the following lemma, which bounds the error in the effective Hamiltonian w.r.t.~to the exact Hamiltonian.

\begin{lemma}[Error bound on effective Hamiltonian]
    Let $H$ be a Hamiltonian and let $U:= \exp \lb -iHt \rb$ with $0 < t \leq \frac{1}{4\norm{H}}$ be the associated time evolution operator. Let $\widetilde{U}$ be a unitary approximation to $U$ such that
    \begin{equation}
        \norm{\widetilde{U} - U} \leq \frac{t\epsilon}{9} \leq \frac{1}{3},
    \end{equation}
    for some error tolerance $\epsilon \geq 0$. Then the effective Hamiltonian $\widetilde{H}$ of $\widetilde{U}$ satisfies
    \begin{equation}
        \norm{\widetilde{H} - H} \leq \epsilon.
    \end{equation}
\label{lem:effective_ham}
\end{lemma}

\begin{proof}
    For $x \in [0,1]$, define
    \begin{equation}
        M(x) := e^{-iH t} + x \Delta, \qquad \Delta := \widetilde{U} - e^{-iH t}.
    \end{equation}
    Furthermore, let
    \begin{equation}
        h(x) := \frac{i}{t} \log M(x), 
    \end{equation}
    where
    \begin{equation}
        \log M(x) = \sum_{k=1}^\infty \frac{(-1)^{k+1}}{k} \lb M(x) - \mathbb{1} \rb^k.
    \end{equation}
    Note that $h(0) = H$ and $h(1) = \widetilde{H}$.
    The above series converges if $\norm{M(x) - \mathbb{1}} < 1$ in the sense that $e^{\log M(x)} = M(x)$. 
    By assumption, we have that $t \leq \frac{1}{4\norm{H}}$ and $\norm{\Delta} \leq 1/3$. Thus,
    \begin{equation}
    \begin{split}
         \norm{M(x) - \mathbb{1}} = \norm{e^{-iH t} + x \Delta - \mathbb{1}} &\leq \norm{e^{-iH t} - \mathbb{1}} + \norm{\Delta} \\
         &\leq \norm{\sum_{k=0}^\infty \frac{(-iH t)^k}{k!} - \mathbb{1}} + \frac{1}{3} \\
         &\leq \norm{\sum_{k=1}^\infty \frac{(-iH t)^k}{k!}} + \frac{1}{3} \\
         &\leq \sum_{k=1}^\infty \frac{t^k}{k!} \norm{H}^k + \frac{1}{3} \\
         &\leq \sum_{k=1}^{\infty} \frac{(1/4)^k}{k!} + \frac{1}{3} \\
         &\leq \frac{1}{1-1/4} - 1 + \frac{1}{3} = \frac{2}{3} < 1,
    \end{split}
    \end{equation}
    as required for convergence.
    
    Next, we use the fundamental theorem of calculus to bound the error in the effective Hamiltonian $\widetilde{H}$ w.r.t.~$H$:
    \begin{equation}
    \begin{split}
        \norm{\widetilde{H} - H} &= \norm{h(1) - h(0)} \\
        &= \norm{\frac{i}{t}\int_0^1 \frac{d}{dx} \log M(x) dx} \\
        &= \norm{\frac{i}{t}\int_0^1 \lb - \sum_{a,b = 0}^{\infty} \frac{(-1)^{a+b+1}}{a+b+1} \lb M(x) - \mathbb{1} \rb^a \frac{dM}{dx} \lb M(x) - \mathbb{1} \rb^b \rb  dx}.
    \end{split}
    \end{equation}
    Using the triangle inequality, the submultiplicativity of the spectral norm, the fact that $\frac{dM}{dx} = \Delta$ and $\norm{\Delta} \leq \frac{t \epsilon}{9}$, we see that
    \begin{equation}
    \begin{split}
          \norm{\widetilde{H} - H} &\leq \frac{\norm{\Delta}}{t} \sum_{a,b=0}^\infty \frac{1}{a+b+1} \max_{x \in [0,1]}\norm{M(x) - \mathbb{1}}^{a+b} \\
         &\leq \frac{\epsilon}{9} \sum_{a,b=0}^\infty \frac{1}{a+b+1} \lb \frac{2}{3} \rb^{a+b} \\
         &\leq \frac{\epsilon}{9} \sum_{a=0}^{\infty} \frac{\lb \frac{2}{3} \rb^a}{a+1}  \sum_{b=0}^\infty \lb \frac{2}{3} \rb^b \\
         &\leq \frac{\epsilon}{9} \frac{1}{(1- 2/3)^2} \\
         &\leq \epsilon.
    \end{split}
    \end{equation}
\end{proof}

The above results allow us to bound the error in the eigenstates obtained via QPE with an imperfect unitary as discussed below.

\begin{prop}
    Let $H \in \mathbb{C}^{N \times N}$ be Hermitian with eigenvalues $\lambda_0 \leq \lambda_1 \leq  \dots \leq \lambda_{N-1}$ and corresponding eigenstates $\ket{\phi_0}, \ket{\phi_1}, \dots, \ket{\phi_{N-1}}$. Fix some eigenstate $\ket{\phi_j}$ and let $\gamma_j := \min \{ |\lambda_j - \lambda_{j-1}|, |\lambda_j - \lambda_{j+1}| \}$ with $\lambda_{-1} := -\infty$ and $\lambda_N := \infty$ denote the minimum gap between $\lambda_j$ and the remaining eigenvalues of $H$. Let $U:= \exp \lb -iHt \rb$ with $0 < t \leq \frac{1}{4\norm{H}}$ be the time evolution operator associated with $H$ and let $\widetilde{U}$ be a unitary approximation to $U$ such that
    \begin{equation}
        \norm{\widetilde{U} - U} \leq \frac{t\epsilon}{9} \leq \frac{1}{3},
    \end{equation}
    for some error tolerance $0 \leq \epsilon \leq \gamma_j/2$. Furthermore, let $\ket{\widetilde{\phi}_j}$ be the $j$-th eigenstate of the effective Hamiltonian $\widetilde{H}$ of $\widetilde{U}$. Then,
    \begin{equation}
        |\braket{\phi_j}{\widetilde{\phi}_j}|^2 \ge 1- \frac{\pi^2 \epsilon^2}{\gamma_j^2}.
    \end{equation}
\label{prop:approximate_overlap}
\end{prop}

\begin{proof}
    From Lemma~\ref{lem:effective_ham} we have that $\norm{\widetilde{H} - H} \leq \epsilon$. By the Davis-Kahan theorem (see Theorem~\ref{thm:davis-kahan}) we then have that
    \begin{equation}
        \sqrt{1-|\braket{\phi_j}{\widetilde{\phi}_j}|^2} \leq \frac{\pi}{2} \frac{\norm{H - \widetilde{H}}}{\delta_j} \leq \frac{\pi}{2} \frac{\epsilon}{\delta_j},
    \end{equation}
    where $\delta_j = \min \{ |\lambda_j - \widetilde{\lambda}_{j-1}|, |\lambda_j - \widetilde{\lambda}_{j+1}| \}$ with $\widetilde{\lambda}_j$ denoting the $j$-th eigenvalue of $\widetilde{H}$. Without loss of generality, assume that $\delta_j = |\lambda_j - \widetilde{\lambda}_{j+1}|$.
    Then we can use Weyl's inequality (see Theorem~\ref{thm:weyl}) to lower bound $\delta_j$ as follows:
    \begin{equation}
    \begin{split}
          \delta_j &= |\lambda_j - \widetilde{\lambda}_{j+1}| = |\lambda_j - \lambda_{j + 1} + \lambda_{j + 1} - \widetilde{\lambda}_{j + 1}| \\
          &\geq \left| |\lambda_j - \lambda_{j+1}| - |\lambda_{j+1} - \widetilde{\lambda}_{j+1}| \right| \\
          &\geq \gamma_j - \epsilon \\
          &\geq \frac{\gamma_j}{2},
    \end{split}
    \end{equation}
    where we used the fact that by assumption, $\epsilon \leq \gamma_j/2$. Putting everything together, we therefore see that
    \begin{equation}
        |\braket{\phi_j}{\widetilde{\phi}_j}|^2 \geq 1 - \frac{\pi^2 \epsilon^2}{\gamma_j^2},
    \end{equation}
    as claimed.
\end{proof}

In the context of our divide-and-conquer strategy for ground state preparation, Proposition~\ref{prop:approximate_overlap} shows that as long as we can implement each $U_s = \exp \lb i\pi H_s/2\norm{H_s} \rb$ within error $\epsilon \in O \lb \gamma_{\min}/H_{\max} \rb$ using only $\mathrm{poly} (N)$ gates, then the overall gate complexity results given in Corollary~\ref{cor:no_van_vleck} remain valid.

Now we are ready to state the main result of this section, which provides a lower bound on the overlap between the eigenstate of some unperturbed Hamiltonian $H_0$ and the corresponding eigenstate of an approximation to the perturbed Hamiltonian $H_0 + H_{\mathrm{int}}$.

\begin{thm}[Lower Bound on Eigenstate Overlap]
    Let $H = H_{0}+ H_{\rm int}$ with $H_0, H_{\rm int} \in \mathbb{C}^{N \times N}$ Hermitian be a Hamiltonian and let $\lambda_0^{(0)} \leq \lambda_1^{(0)} \leq  \dots \leq \lambda_{N-1}^{(0)}$ be the eigenvalues of $H_0$ with associated eigenstates $\ket{\phi_0^{(0)}}, \ket{\phi_1^{(0)}}, \dots, \ket{\phi_{N-1}^{(0)}}$. Further, for any $j \in \{ 0, 1, \dots, N-1\}$, let $\gamma_j^{(0)} := \min \{ |\lambda_j^{(0)} - \lambda_{j-1}^{(0)}|, |\lambda_j^{(0)} - \lambda_{j+1}^{(0)}| \}$ with $\lambda_{-1}^{(0)} := -\infty$ and $\lambda_N^{(0)} := \infty$ denote the minimum gap between $\lambda_j^{(0)}$ and the remaining eigenvalues of $H_0$.
    Assume that we have access to a unitary operation $\widetilde{U}$ such that $\|\widetilde{U}- e^{-iHt}\| \le \frac{t \epsilon}{9}$ with $0 < t \leq \frac{1}{4\norm{H}}$ and $0 \leq \epsilon \leq \norm{H_{\mathrm{int}}}$. Let $\widetilde{\lambda}_0 \leq \widetilde{\lambda}_1 \leq  \dots \leq \widetilde{\lambda}_{N-1}$ be the eigenvalues of the effective Hamiltonian $\widetilde{H}$ of $\widetilde{U}$ according to Definition~\ref{def:effective_ham} with associated eigenstates $\ket{\widetilde{\phi}_0}, \ket{\widetilde{\phi}_1}, \dots, \ket{\widetilde{\phi}_{N-1}}$.
    If $\gamma_j^{(0)} > 2\norm{H_{\mathrm{int}}}$, then
    \begin{equation}
        |\braket{\widetilde{\phi}_j}{\phi_j^{(0)}}|^2 \ge 1- \frac{\pi \|H_{\rm int}\|}{\gamma_j^{(0)} - 2\|H_{\rm int}\|}.
    \end{equation}
\label{thm:perturbed_overlap}
\end{thm}

\begin{proof}
    From Lemma~\ref{lem:effective_ham} we have that $\norm{\widetilde{H} - H} \leq \epsilon$. It then follows from the Davis-Kahan theorem (see Theorem~\ref{thm:davis-kahan}) and the triangle inequality that
    \begin{equation}
        \sqrt{1 - |\braket{\widetilde{\phi}_j}{\phi_j^{(0)}}|^2} \leq \frac{\pi}{2} \frac{\norm{\widetilde{H} - H_0}}{\delta_j} \leq \frac{\pi}{2} \frac{\norm{H_{\mathrm{int}}} + \epsilon}{\delta_j} \leq \frac{\pi \norm{H_{\mathrm{int}}}}{\delta_j},
    \end{equation}
    where $\delta_j = \min \{ | \lambda_j^{(0)} - \widetilde{\lambda}_{j-1}|, | \lambda_j^{(0)} - \widetilde{\lambda}_{j+1}| \}$ with $\widetilde{\lambda}_{-1} := -\infty$ and $\widetilde{\lambda}_N := \infty$. Using the reverse triangle inequality, we see that
    \begin{equation}
        | \lambda_j^{(0)} - \widetilde{\lambda}_{j\pm 1}| \geq \left| | \lambda_j^{(0)} - \lambda_{j\pm 1}^{(0)}| - | \lambda_{j\pm 1}^{(0)} - \widetilde{\lambda}_{j\pm 1}| \right|.
    \end{equation}
    By definition, the first term on the RHS satisfies $| \lambda_j^{(0)} - \lambda_{j\pm 1}^{(0)}| \geq \gamma_j^{(0)}$. Next, we use Weyl's inequality to conclude that the second term on the RHS is upper bounded as follows:
    \begin{equation}
        | \lambda_{j\pm 1}^{(0)} - \widetilde{\lambda}_{j\pm 1}| \leq \norm{H_0 - \widetilde{H}} \leq \norm{H_{\mathrm{int}}} + \epsilon \leq 2\norm{H_{\mathrm{int}}}.
    \end{equation}
    By assumption, $\gamma_j^{(0)} > 2\norm{H_{\mathrm{int}}}$. Thus,
    \begin{equation}
        | \lambda_j^{(0)} - \widetilde{\lambda}_{j\pm1}| \geq \gamma_j^{(0)} - 2\norm{H_{\mathrm{int}}}.
    \end{equation}
    Putting everything together, we therefore have that
    \begin{equation}
        1 - |\braket{\widetilde{\phi}_j}{\phi_j^{(0)}}|^2 \leq \sqrt{1 - |\braket{\widetilde{\phi}_j}{\phi_j^{(0)}}|^2} \leq \frac{\pi \norm{H_{\mathrm{int}}}}{\gamma_j^{(0)} - 2\norm{H_{\mathrm{int}}}},
    \end{equation}
    which implies that
    \begin{equation}
        |\braket{\widetilde{\phi}_j}{\phi_j^{(0)}}|^2 \geq 1 - \frac{\pi \|H_{\rm int}\|}{\gamma_j^{(0)} - 2\|H_{\rm int}\|}.
    \end{equation}
\end{proof}

In the context of ground state preparation, the above theorem states that
\begin{equation}
    |\braket{\widetilde{\phi}_0}{\phi_0^{(0)}}|^2 \geq 1 - \frac{\pi \|H_{\rm int}\|}{\gamma_0^{(0)} - 2\|H_{\rm int}\|},
\end{equation}
assuming that $\gamma_0^{(0)} > 2\norm{H_{\mathrm{int}}}$, where $\gamma_0^{(0)}$ is the spectral gap of $H_0$.
This shows that, in the event that the gap is large compared to the strength of the interactions that we are considering when combining subsystems, the success probability of phase estimation will be large.  
However, even if the interaction strength between subsystems is weak, challenges can emerge in the limit of large $N$. In particular, for a system with $N$ weakly interacting subsystems with $N$ interaction Hamiltonians $H_{{\rm int},j}$, the spectral gap $\gamma(n)$ after having added the first $n$ interaction terms obeys
\begin{equation}
    \gamma(n) \ge \max\{\gamma_0^{(0)} - 2\sum_{j=0}^{n-1}\|H_{{\rm int},j}\|,0\}.
\label{gap_lower_bound}
\end{equation}
Thus, without further assumptions, if the strength of the interactions are bounded below by a constant, then our lower bound on the gap vanishes as the number of subsystems increases. From Theorem~\ref{thm:perturbed_overlap} we see that our lower bound for the success probability of phase estimation, which depends on the overlap, also goes to~0 since the lower bound on the overlap obeys
\begin{equation}
    r^2 \ge \max \left\{1 - \frac{\pi \max_j\|H_{{\rm int},j}\|}{\max\{\gamma_0^{(0)} - 2\sum_{j=0}^{N-1}\|H_{{\rm int},j}\|,0\}},0 \right\},
\label{overlap_lower_bound}
\end{equation}
which approaches $0$ as $N \rightarrow \infty$ if $\|H_{{\rm int},j}\|$ is bounded below by a constant.

There are a number of important points about such an argument. The first is that if we are in a situation where the collective interaction strength between the systems is so strong that the gap could be closed by the strength of the interactions, then we are in a regime where the intermolecular effects are potentially dominating the chemical properties of the individual systems.  In such situations, it would not be appropriate to model the system as a set of weakly interacting subsystems but rather as a single collective system. 

Secondly, it is important to note that in the above discussion we are concatenating two lower bounds in order to find a lower bound $r$ on the overlap between the ground states of adjacent layers in the binary tree decomposition of the Hamiltonian. However, neither the lower bound on the overlap derived in Theorem~\ref{thm:perturbed_overlap} nor the subsequent lower bound on the spectral gap in Eq.~\eqref{gap_lower_bound} are necessarily tight. This means that the resulting lower bound given in Eq.~\eqref{overlap_lower_bound} can be overly pessimistic.

Furthermore, there is an important feature that the preceding analysis avoids discussing which is the impact of the nature of the state vector on the perturbed state itself.  This shift is not conveniently described by the Davis-Kahan theorem but can be understood in terms of perturbation theory.  Specifically, let $H(\tau) = H_{0} + \tau H_{\rm int}$ for $\tau \in [0,1]$ and define $\ket{\phi_j(\tau)}$ to be the instantaneous eigenvectors at time $\tau$ and similarly define $E_j(\tau)$ to be the corresponding eigenvalues. Then
\begin{equation}
    \partial_\tau \ket{\phi_j(\tau)} = \sum_{k\ne j}\frac{\ketbra{\phi_k(\tau)}{\phi_k(\tau)} H_{\rm int} \ket{\phi_{j}(\tau)} }{E_{j}(\tau)-E_{k}(\tau)}.
\end{equation}
Thus we have from the mean value theorem (under the assumption that the derivative is defined over the range) that for each $\ket{\phi_\ell}=\ket{\phi_\ell(0)}$ there exists a value $\tau_\ell\in [0,1]$ such that
\begin{equation}
    \braket{\phi_\ell}{\phi_j(1)} = \delta_{\ell,j} + \sum_{k\ne j}\frac{\braket{\phi_\ell}{\phi_k(\tau_\ell)}\bra{\phi_k(\tau_\ell)} H_{\rm int} \ket{\phi_{j}(\tau_\ell)} }{E_{j}(\tau_\ell)-E_{k}(\tau_\ell)}.
\end{equation}
Using the Cauchy-Schwarz inequality, we then obtain
\begin{equation}
    |  \braket{\phi_\ell}{\phi_j(1)} - \delta_{\ell,j}|\le \max_{\tau_\ell}\frac{\sqrt{\bra{\phi_j(\tau_\ell)} H_{\rm int}^2 \ket{\phi_j(\tau_\ell)}}}{\min_{k\ne j} |E_{j}(\tau_\ell)-E_{k}(\tau_\ell)|}.
\end{equation}
From this we see that the perturbed overlap is given by the maximum matrix element along the path.  If the interaction Hamiltonian has a small expectation value in the ground state of the non-interacting Hamiltonian and intermediate eigenstates, then the success probability may differ greatly from the previous estimates which make no assumptions about the inner products of the terms.  In Section~\ref{sec:numerics}, we will see that in practice there can be a substantial gap between the worst case bounds and the empirical performance, which suggests that state-dependent effects can be significant in practice.

\subsection{Sufficient Conditions for Quasi-Polynomial Running Time}

In the following, we will discuss some sufficient conditions on the system Hamiltonian that ensure a quasi-polynomial running time for preparing the ground state via the divide-and-conquer strategy discussed in Theorem~\ref{thm:divide+conquer}. Let us again consider a system composed of $N=2^p$ basic subsystems and let the overall Hamiltonian $H$ have a binary tree decomposition as described in Definition~\ref{def:input_ham}. Consider a collection of subsystems labeled by a binary string $s$ of length $0 \leq k \leq p-1$ and let $\ket{\psi_s}$ be the ground state of the associated Hamiltonian $H_s$. Further, let $s0$ and $s1$ be the binary strings of length $k+1$ that label the child nodes of $s$ and let $\ket{\psi_{s0}}$ and $\ket{\psi_{s1}}$ be the corresponding ground states. It then follows from Theorem~\ref{thm:perturbed_overlap} that
\begin{equation}
    \left|\bra{\psi_s} \lb \ket{\psi_{s0}}\ket{\psi_{s1}} \rb \right|^2  \ge 1- \frac{\pi \|A_{s}\|}{\widetilde{\gamma}_{s} - 2 \|A_{s}\|} =1- \frac{\pi/2}{\widetilde{\gamma}_{s}/2\|A_{s}\| - 1},
\end{equation}
where $\widetilde{\gamma}_s$ is the spectral gap of $H_s - A_s = H_{s0} + H_{s1}$. Equivalently, $\widetilde{\gamma}_s = \min \left\{ \gamma_{s0}, \gamma_{s1} \right\}$ where $\gamma_{s0}$ is the spectral gap of $H_{s0}$ and $\gamma_{s1}$ is the spectral gap of $H_{s1}$. 
Now, as long as $\widetilde{\gamma}_s/\norm{A_s} > 1 + \pi/2$ for all labels $s \in \left\{ \{ 0, 1\}^j | j \in \{0, 1, \dots, p-1\} \right\}$ of the binary tree, we obtain asymptotically better scaling than what would be expected from the Van Vleck argument since $r$ is constant in this case. In particular, as shown in Corollary~\ref{cor:no_van_vleck}, the gate complexity in this case scales like $O \lb N^{\log \log \lb N \rb} \mathrm{poly} (N) \rb$ which is quasi-polynomial in the number of basic subsystems.

The following lemma discusses more specific sufficient conditions for ensuring that $r$ is constant in the system size and hence allowing for a quasi-polynomial upper bound on the overall running time.

\begin{lemma}[Sufficient Conditions for Constant Overlap]
    Consider a system composed of $N=2^p$ basic subsystems and let the overall Hamiltonian $H$ have a binary tree decomposition as described in Definition~\ref{def:input_ham}. Let $k \in \{0, 1, \dots, p-1\}$ label any layer of the binary tree and let $A_{k}^{\max} := \max_{s \in \{0,1\}^k} \left\{ \norm{A_s} \right\}$ be an upper bound on the norm of the interaction terms in layer $k$. Further, let $\gamma_{p}^{\min}$ be the minimum spectral gap over all $2^p$ basic subsystem Hamiltonians and let $c > 1 + \pi/2$ be a constant.
    Then for any bitstring label $s$ of the binary tree, the ground state overlap $\left|\bra{\psi_s} \lb \ket{\psi_{s0}}\ket{\psi_{s1}} \rb \right|$ is lower bounded by a positive constant if for all $k$,
    \begin{equation}
        A_{k}^{\max} \leq \frac{1}{2c} \lb \gamma_{p}^{\min} - 2 \sum_{j=k+1}^{p-1} A_{j}^{\max} \rb.
    \end{equation}
    This is satisfied if for all $k$,
    \begin{equation}
        A_{k}^{\max} \leq \frac{\gamma_{p}^{\min}}{4c(p-k)^2}.
    \end{equation}
\end{lemma}

\begin{proof}
    Let $\gamma_{k}^{\min} := \min_{s \in \{0,1\}^k} \left\{ \gamma_s \right \}$ be a lower bound on the spectral gap of any $H_s$ in layer $k$ and let $\widetilde{\gamma}_{k}^{\min} := \min_{s \in \{0,1\}^k} \left\{ \widetilde{\gamma}_s \right \}$ be a lower bound on the spectral gap of any $H_s - A_s$ in layer $k$. According to these definitions, $\widetilde{\gamma}_{k}^{\min} = \gamma_{k+1}^{\min}$. Then, 
    \begin{equation}
    \begin{split}
         \widetilde{\gamma}_{k}^{\min} = \gamma_{k+1}^{\min} \geq \widetilde{\gamma}_{k+1}^{\min} - 2 A_{k+1}^{\max},
    \end{split}
    \end{equation}
    which follows directly from Weyl's inequality (see Theorem~\ref{thm:weyl}). Unrolling the recurrence relation, we find that
    \begin{equation}
         \widetilde{\gamma}_{k}^{\min} \geq \gamma_{p}^{\min} - 2 \sum_{j=k+1}^{p-1} A_{j}^{\max},
    \end{equation}
    where $\gamma_{p}^{\min}$ denotes the minimum spectral gap over all $2^p$ basic subsystem Hamiltonians (corresponding to the leaf nodes of the binary tree decomposition of $H$). Note that any sum with a lower summation limit that is larger than the upper summation limit is set to $0$.
    The overlap $\left|\bra{\psi_s} \lb \ket{\psi_{s0}}\ket{\psi_{s1}} \rb \right|$ is thus lower bounded by a positive constant if there exists a constant $c > 1 + \frac{\pi}{2}$ such that
    \begin{equation}
    \begin{split}
        \frac{\widetilde{\gamma}_s}{2\norm{A_s}} &\geq \frac{\widetilde{\gamma}_{k}^{\min}}{2A_{k}^{\max}} \\
        &\geq \frac{\gamma_{p}^{\min} - 2 \sum_{j=k+1}^{p-1} A_{j}^{\max}}{2A_{k}^{\max}}. \\
        &\geq c > 1 + \frac{\pi}{2}.
    \end{split}
    \end{equation}
    This implies that it suffices to ensure that
    \begin{equation}
        A_{k}^{\max} \leq \frac{1}{2c} \lb \gamma_{p}^{\min} - 2 \sum_{j=k+1}^{p-1} A_{j}^{\max} \rb,
    \label{bound_max_int}
    \end{equation}
    as claimed.
    
    Suppose now that $A_{k}^{\max} \leq \frac{\gamma_{p}^{\min}}{4c(p-k)^2}$ for all $k \in \{0, 1, \dots, p-1\}$ with $c > 1 + \pi/2$ as before. Then the RHS of Eq.~\eqref{bound_max_int} satisfies
    \begin{equation}
    \begin{split}
          \gamma_{p}^{\min} - 2 \sum_{j=k+1}^{p-1} A_{j}^{\max} &\geq \gamma_{p}^{\min} - 2 \sum_{j=k+1}^{p-1} \frac{\gamma_{p}^{\min}}{4c(p-j)^2} \\
          &\geq \gamma_{p}^{\min} - \frac{\gamma_{p}^{\min}}{2c} \sum_{j=1}^{p-k-1} \frac{1}{j^2} \\
          &\geq \gamma_{p}^{\min} - \frac{\gamma_{p}^{\min}}{2c} \sum_{j=1}^{\infty} \frac{1}{j^2} = \gamma_{p}^{\min} - \frac{\gamma_{p}^{\min} \pi^2}{12c} \\
          &\geq \frac{2\gamma_{p}^{\min}}{3}.
    \end{split}
    \end{equation}
    Therefore,
    \begin{equation}
        A_{k}^{\max} \leq \frac{\gamma_{p}^{\min}}{4c(p-k)^2} \leq \frac{\gamma_{p}^{\min}}{4c} \leq \frac{\gamma_{p}^{\min}}{3c} \leq \frac{1}{2c} \lb \gamma_{p}^{\min} - 2 \sum_{j=k+1}^{p-1} A_{j}^{\max} \rb,
    \end{equation}
    as required.
\end{proof}

A physical Hamiltonian that satisfies the condition $A_{k}^{\max} \leq \frac{\gamma_{p}^{\min}}{4c(p-k)^2}$ might be a system composed of $N~=~2^p$ well-separated molecules which are weakly interacting with each other such that their interactions are dominated by the dipole contribution of the multipole expansion of the Coulomb potential.
If that is the case, then the overlap between the ground states of adjacent layers of the binary tree would be lower bounded by a constant, i.e.~$r \in \Theta(1)$. According to Corollary~\ref{cor:no_van_vleck}, if additionally, $H_{\max}/\gamma_{\min} \in O \lb \mathrm{poly} (N) \rb$ and each $V_0, V_{1} , \dots V_{N-1}$ and each $U_s$ can be implemented with gate complexity in $O \lb \mathrm{poly} (N) \rb$, then the overall gate complexity would scale like $O \lb N^{\log \log \lb N \rb} \mathrm{poly} (N) \rb$ which is quasi-polynomial in the number of molecules.

\subsection{Area Laws vs Volume Laws}

Above we showed that the Van Vleck catastrophe can be averted if the success probability for the ``conquer step'', which involves the projection of the product state of the ground states of the subsystems onto the ground state of the combined system, is sufficiently large. We will see here that the maximum success probability of this step is intimately related to the entanglement structure of the ground state of the interacting Hamiltonian.  To get an idea about this, let us examine the case where the entanglement entropy is maximal across a given bipartition of a system $C = A \cup B$. Specifically, let $\rho$ be the density matrix of the ground state of the entire system $C$ and let $\ket{\psi_A}$ and $\ket{\psi_B}$ be the ground states of subsystems $A$ and $B$, respectively. Additionally, let $\rho_A := \mathrm{Tr}_B \rho$ and $\rho_B := \mathrm{Tr}_A \rho$ denote the reduced density matrices of $\rho$ on subsystems $A$ and $B$. The entanglement entropy associated with the bipartition is then given by
\begin{equation}
    S(\rho_A) = - \mathrm{Tr} \lb \rho_A \log \rho_A \rb = - \mathrm{Tr} \lb \rho_B \log \rho_B \rb = S(\rho_B).
\end{equation}
Now, let $\dim_A$ and $\dim_B$ denote the dimensions of subsystems $A$ and $B$ and assume that $\dim_A \leq \dim_B$. Then the entanglement entropy is maximized if $\rho_A$ is the maximally mixed state, i.e. $\rho_A = \frac{\mathbb{1}}{\dim_A}$.
In this case, the probability of success for our divide-and-conquer strategy is upper bounded as follows:
\begin{equation}
    P_{\textrm{succ}}={\rm Tr}(\rho \ketbra{\psi_A}{\psi_A} \otimes \ketbra{\psi_B}{\psi_B}) \le \bra{\psi_A}\rho_A\ket{\psi_A} = \frac{1}{\dim_A},
\end{equation}
This shows that if we have too much entanglement entropy, then the divide-and-conquer strategy will fail.  On the other hand, if $\rho$ is a product state then the success probability $P_{\textrm{succ}}$ could be as large as $1$ or as small as $0$. The Davis-Kahan theorem shows that, in the case where the spectral gap is large enough relative to the strength of the interaction Hamiltonian, we expect the overlap to be close to $1$; however, here we wish to examine the question of when the qualitative scaling of the entanglement forbids the divide and conquer algorithm from succeeding and we will phrase this characterization in terms of entropy area and volume laws. 

Volume and area laws for entanglement provide an important characterization of the structure of Hamiltonians.  At a high level, an area law for entanglement exists if for all partitions of the space $C$ into disjoint subsystems $A,B$ such that $C = A \cup B$, the entanglement of the bipartition scales with the area of the boundary that separates systems $A$ and $B$. In particular, a pure state $\rho$ with reduced density matrix $\rho_A$ is said to have an area or volume law for entanglement if~\cite{hastings2007area}
\begin{align}
    \text{Area Law for Entanglement}~&\Rightarrow S(\rho_A) \in O(|{\rm bdy}(A)|)\nonumber\\
    \text{Volume Law for Entanglement}~&\Rightarrow S(\rho_A) \in O(|A|), \nonumber
\end{align}
where $|A|$ is the number of qubits (or more generally qudits of finite dimension) within subsystem $A$ and $|{\rm bdy}(A)|$ is the number of qubits at the boundary separating $A$ from $B$. Note that because of the symmetry of the entanglement entropy, the above definition holds without modification if the roles of $A$ and $B$ are switched.
Volume laws tend to be ubiquitous for eigenstates of random Hamiltonians~\cite{vidmar2017entanglement}, whereas area laws are much rarer. They do occur in general for one-dimensional lattice Hamiltonians~\cite{hastings2007area}, but the conditions required for higher-dimensional area laws remain an active area of research~\cite{eisert2008area}.

Now we wish to turn our attention to the question of what the maximum success probability is that could be sustained for a fixed amount of entanglement.  This is important because if we have a system with a volume law of entanglement then the resulting reduced density matrices may be too mixed in order to support a high probability of success. Specifically, let us consider the partial success probability ${\rm Tr}(\rho (\ketbra{\psi_A}{\psi_A} \otimes \openone))$ and assume it is equal to $r_A^2 \in [0,1]$, i.e.,
\begin{equation}
    {\rm Tr}(\rho  (\ketbra{\psi_A}{\psi_A} \otimes \openone)) = {\rm Tr}(\rho_A \ketbra{\psi_A}{\psi_A})=r^2_A.
\end{equation}
Note that $r_A^2$ is an upper bound on the overall success probability $P_{\textrm{succ}}$.
It is straightforward to verify that the following reduced density matrix maximizes the von Neumann entropy for fixed $r_A^2$:
\begin{equation}
    \rho_A = r_A^2 \ketbra{\psi_A}{\psi_A} + \frac{1-r_A^2}{\dim_A-1} (\openone - \ketbra{\psi_A}{\psi_A}).
\end{equation}
The von Neumann entropy of this state is then
\begin{equation}
    S(\rho_A) = -r_A^2 \log(r_A^2) - (1-r_A^2) \log\left(\frac{1-r_A^2}{\dim_A -1} \right) \geq - (1-r_A^2) \log\left(\frac{1-r_A^2}{\dim_A -1} \right).
\end{equation}
Next, let us argue about the achievability of a given level of overlap for a fixed level of entanglement entropy. If we wish to hit a fixed value of entanglement entropy, $E$, for the pure state $\rho$ then
\begin{equation}
    E \ge - (1-r_A^2) \log\left(\frac{1-r_A^2}{\dim_A -1} \right). 
\end{equation}
As the entanglement entropy is a decreasing function of $r_A^2$, for $r_A^2 \ge 1/\dim_A$, we have that a value of $r_A^2$ always exists such that the above inequality is satisfied.  However, the largest value of $r_A^2$ that can be attained can be found by solving the expression for $r_A^2$ which yields
\begin{equation}
    r_A^2 = 1 + \frac{E}{W(-E/(\dim_A-1))} = 1 - \frac{E}{\log(\dim_A/E) -\log\log(\dim_A/E)+o(1)},
\end{equation}
where $W$ is the Lambert-W function.
Thus, 
\begin{equation}
    |1-r_A^2| \in {\Omega}\left(\frac{E}{\log(\dim_A/E)} \right).
\end{equation}

Now let us examine what this means in terms of area vs volume law scaling.  If we assume that we have a lattice Hamiltonian then the maximum entanglement across any bipartition satisfies an area law in $d$ spatial dimensions with $E \in O(\log^{1-1/d}(\dim_A))$.  In this case we have that
\begin{equation}
    |1-r_A^2| \in {\Omega}\left(\frac{\log^{1-1/d}(\dim_A)}{\log(\dim_A)} \right),
\end{equation}
which vanishes as the dimension of the Hilbert space tends to infinity.  This means that an area law for entanglement does not necessarily preclude a high probability of success in our divide and conquer algorithm and considerations of the eigenvalue gaps will be needed to assess whether the algorithm will be capable of succeeding with high probability.

If we assume, on the other hand, that there is a volume law scaling for the entanglement entropy then  $S(\rho_A)= (1-\eta)\log(\dim_A)$ for constant $0\le \eta\le 1$. Hence choosing $E$ to saturate this value yields,
\begin{align}
    r_A^2 &= 1-\frac{(1-\eta)\log(\dim_A)}{\log(\dim_A) - \log(1-\eta) - O\left(\log(\log(\dim_A)/(1-\eta)\log(\dim_A))\right) }\nonumber\\
    &= 1 -(1-\eta)\left(1+\frac{\log(1-\eta)}{\log(\dim_A)} + O\left(\log\log\left(\frac{1-\eta}{\log(\dim_A)}\right)\right) \right)\nonumber\\
    &=\eta - \frac{(1-\eta)\log(1-\eta)}{\log(\dim_A)} + O\left(\frac{(1-\eta)\log^2(1-\eta)}{\log^2(\dim_A)} \right).
\end{align}
A consequence of this is that the question of whether high success is even possible for volume law scalings can actually be answered in cases where $\eta\in o(1)$.  In such cases, we see that $r_A^2 \in o(1)$ as well which suggests that high success probability is impossible for such volume law scalings.  This is indeed the case for random Hamiltonians drawn, for example, from the Gaussian Unitary Ensemble~\cite{vidmar2017entanglement} where $E = \log(\dim_A)+ O(\dim_A^2/2^n)$ where $2^n$ is the Hilbert space dimension that operators in the larger system $C$ act on.  Thus in this case, large success probability using a divide and conquer scheme is provably impossible.  However, if a more modest volume law scaling with $\eta \in \Theta(1)$ is achievable, then even a volume law scaling cannot be specifically excluded from the conditions required by Corollary~\ref{cor:no_van_vleck} to ensure quasi-polynomial scaling of the divide and conquer state preparation algorithm. However, just as argued in the area law case, many of these cases may further be unachievable after considering the spectral gap and the strength of the perturbing Hamiltonian.

\section{Numerics}
\label{sec:numerics}

The usefulness of Theorem~\ref{thm:divide+conquer} hinges on a specific assumption: the system of interest must ensure a finite overlap between the ground states of adjacent layers in the binary tree when increasing the system size. Although this assumption does not hold in general, it is possible to satisfy it in certain cases. To demonstrate this numerically, we consider the 1D Transverse Field Ising Model which plays an important role in studying, e.g., phase transitions~\cite{pfeuty1970one}. Its simple structure makes it an ideal candidate for our method. The Hamiltonian of the 1D Transverse Ising Model without periodic boundary conditions is defined as
\begin{equation}
\label{eq:heisenberg}
    H_\mathrm{TFIM} := H_0 + H_\mathrm{int} = h \sum_{i=0}^{N-1} \sigma_{i}^{z} + J \sum_{i=0}^{N-2} \sigma_{i}^x \sigma_{i+1}^x \,,
\end{equation}
where $\sigma_{i}^{x}$ ($\sigma_i^z$) denotes the $X$-Pauli ($Z$-Pauli) matrix on spin $i$, $h$ is the on-site interaction and $J$ the coupling constant between two spins. Note that the ground state of the transverse Ising model can be found analytically by transforming to free fermions; however, for our purposes this is not problematic as it allows us to easily validate the numerical results that are returned by our benchmarks.

In this numerical study, we focus on Ising chains with $2^p$ spins and the on-site interaction equaling the coupling constant, $h=J=1$. We construct the binary tree by starting at the lowest level, $j=p$, with individual spins, each defined by the on-site Hamiltonian $H_0$. The subsequent level, $j=p-1$, includes two spins and the interaction term, $H_\mathrm{int}$, between them. We iteratively continue this process to build the complete binary tree up to the root node with $j=0$.

Although solving a single 1D Transverse Field Ising Model (TFIM) analytically is straightforward, the combination of two 1D TFIMs of two nodes in the binary tree introduces two distinct sets of fermionic creation and annihilation operators, complicating the computation of their overlap. We therefore employ tensor networks and DMRG computations readily implemented in the software library ITensor~\cite{itensor} to determine the ground state of the entire system (at $j=0$) and the ground states at various levels of the binary tree. To ensure the convergence of ground state energies, we verify the energy convergence with respect to the bond dimension and set a maximal allowed bond dimension of $M=20000$ for all calculations. In Fig.~\ref{fig:heisenberg} (left), we illustrate the overlap between the ground states of two levels in the binary tree, $\left| \bra{\psi_{s}}(\ket{\psi_{s0}}\ket{\psi_{s1}}) \right|$, as well as the overlap between the ground state of the full system and the ground state at the lowest level of the tree. As can be seen, a naive preparation of the single ground state of each system (a single spin) would inevitably lead to the undesirable Van Vleck catastrophe. In contrast, our tree-based method offers a more efficient solution given that the overlap between any two levels remains effectively constant, thereby circumventing the Van Vleck catastrophe. In Fig.~\ref{fig:heisenberg} (right), we furthermore show the energy gap between the ground state and the first excited state at the various levels of the binary tree, as well as the norm of the interaction between the two child nodes of a given node in the binary tree. These numerical results show that our analytical lower bound provided in Theorem~\ref{thm:perturbed_overlap} can be loose in practice and that it may be possible to apply our method more broadly than our analytical bound suggests.  
\begin{figure}
    \centering
    \begin{minipage}{0.499\textwidth}
        \centering
        \includegraphics[width=0.9\textwidth]{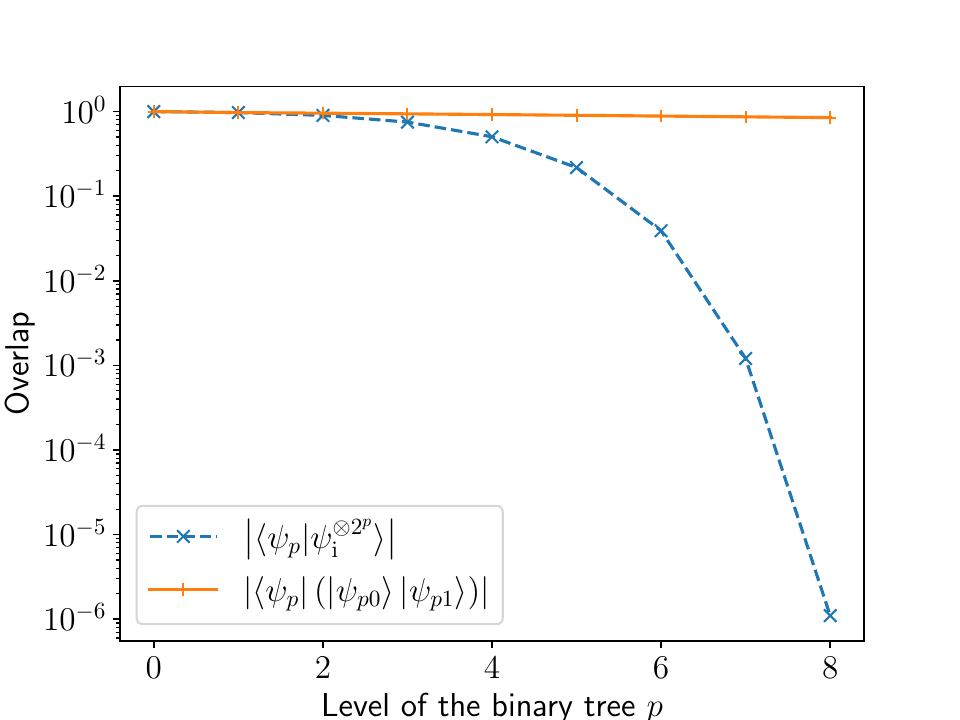} 
    \end{minipage}\hfill
    \begin{minipage}{0.499\textwidth}
        \centering
        \includegraphics[width=0.9\textwidth]{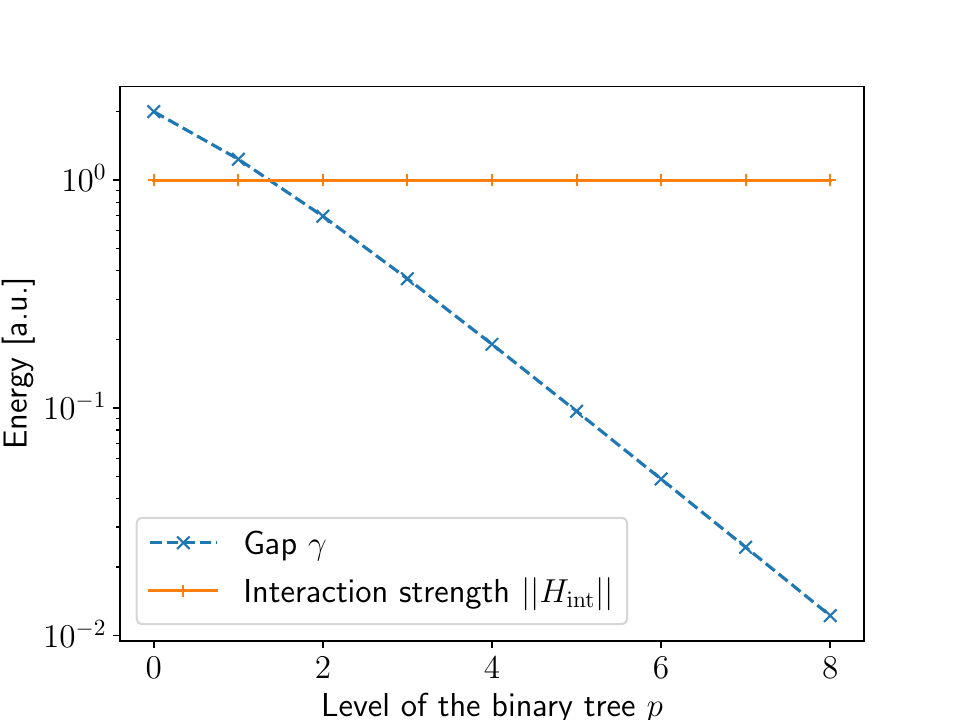} 
    \end{minipage}
    \caption{(left) Comparison of the ground state overlaps for a 1D Transverse Ising Hamiltonian of size $2^p$, see Eq.~(\ref{eq:heisenberg}). The dashed blue line represents the overlap between the Hamiltonian's ground state and the tensor product of $2^p$ individual single-spin ground states, while the solid orange line illustrates the overlap with the tensor product of the two ground states at size $2^{p-1}$. (right) The spectral gap $\gamma$ and the interaction strength $||H_\mathrm{int}||$ for a 1D TFIM model with $2^p$ spins.
    }
    \label{fig:heisenberg}
\end{figure}

It is important to mention that a different choice of the initial state could also mitigate the Van Vleck catastrophe. However, this example underscores that a naive initial state can lead to an exponentially vanishing overlap with the target ground state, and hence an exponential cost for projector-based ground state preparation methods~\cite{Ge2019, Lin2020}, whereas our approach enables the preparation of the target ground state in quasi-polynomial time even with a naive initial state.

\section{Divide and Conquer State Preparation for Fermions }
\label{sec:fermion}

In the previous sections, we discussed the Van Vleck catastrophe and its implications for general quantum systems, supplemented by a numerical example involving a 1D TFIM. 
A major area of interest of quantum computing, however, lies in its application to solving quantum chemistry problems, specifically the computation of ground state energies of interacting fermions. This introduces some subtle challenges, particularly due to the delicate nature of working with fermions. 

As a consequence of the Pauli exclusion principle, fermionic wave functions must be antisymmetric with respect to the exchange of any two electrons. This requirement complicates the application of our divide-and-conquer strategies to a collection of, e.g., $N$ molecules. Initially, one might consider treating each of the $N$ molecules as a distinct subsystem and follow our procedure outlined in the previous sections. However, this approach encounters a significant hurdle: while the ground states of individual subsystems are antisymmetrized internally, their combination does not maintain this property. More critically, the tensor product of two fermionic Hilbert spaces does not yield another fermionic space. When representing the subsystems in first quantization, this issue can theoretically be mitigated by applying a projector onto the antisymmetric subspace at each node, albeit with a success probability that vanishes exponentially with the system size. As we demonstrate, a more viable strategy involves adopting a second quantized representation.

Rather than starting with segmenting the system based on groups of molecules, we first construct a proper basis for the entire system by defining creation and annihilation operators with the following anti-commutation relations:
\begin{align}
    \{a_p, a^\dagger_q\} &= a_p a^\dagger_q +a^\dagger_q a_p = \delta_{pq}, \\
    \{a^\dagger_p, a^\dagger_q \} &= \{a_p, a_q\} = 0.
\end{align}
In quantum chemistry, this can be achieved by performing a Hartree-Fock calculation or building orthogonal orbitals from atomic orbitals. After this procedure, we find the fermionic Hamiltonian in second quantization of the entire system: 
\begin{equation} 
    H = \sum_{pq} T_{pq} a^\dagger_p a_q + \sum_{pqrs} V_{pqrs} a^\dagger_p a^\dagger_q a_r a_s\,,
\end{equation}
where $T_{pq}$ and $V_{pqrs}$ denote the one- and two-body coefficients. As a single orbital is not a function of a single molecule but of the entire system, this formulation does not allow for a division of orbitals that aligns with the molecular segmentation into $N=2^p$ molecules. We note, however, that employing localized orbitals as a workaround enables a closer alignment of these two perspectives. Consequently, our divide-and-conquer technique is applied to a partition of the $2^p$ orbitals, which describe the Hamiltonian, rather than the molecules themselves. 

Simulating the fermionic Hamiltonian on a quantum computer requires the translation of fermionic operators into qubit operators, which can be accomplished through the Jordan-Wigner transformation, expressed as follows: 

\begin{align}
    a^\dagger_p &\rightarrow \lb \prod_{j=0}^{p-1} \bigotimes Z_j \rb \otimes \frac{1}{2} \lb X_p + iY_p \rb \\
    a_p &\rightarrow \lb \prod_{j=0}^{p-1} \bigotimes Z_j \rb \otimes \frac{1}{2} \lb X_p - iY_p \rb.
\end{align}

We highlight that this transformation retains locality in the following sense. If $p < q$ we have that
\begin{equation}
    a^\dagger_p a_q = \frac{1}{4} \lb X_p + iY_p \rb Z_p \otimes \lb \prod_{j=p+1}^{q-1} \bigotimes Z_j \rb \otimes \lb X_q - iY_q \rb,
\end{equation}
which acts nontrivially only on orbitals within the index range $[p, q]$.
Similarly, if $p > q$, then
\begin{equation}
    a^\dagger_p a_q = \frac{1}{4} Z_q\lb X_q + iY_q \rb \otimes \lb \prod_{j=q+1}^{p-1} \bigotimes Z_j \rb \otimes \lb X_p - iY_p \rb,
\end{equation}
which acts nontrivially only on orbitals within the index range $[q, p]$.
Lastly, if $p=q$ we have
\begin{equation}
    a^\dagger_p a_p = \frac{1}{2} \lb \mathbb{1}_p + Z_p\rb.
\end{equation}
This shows that the one-body operators exhibit a certain type of locality. The same is actually true for the two-body terms $a^\dagger_p a^\dagger_q a_r a_s$. Specifically, let $i_{\min} := \min \{ p, q, r,s \}$ and $i_{\max} := \max \{ p, q, r,s \}$. Then $a^\dagger_p a^\dagger_q a_r a_s$ acts nontrivially only on orbitals within the index range $[i_{\min}, i_{\max}]$.
This localization allows for a strategy to label the orbitals, ensuring that qubit operators remain local within their respective nodes in the binary tree structure. By organizing the binary tree such that the Hamiltonian of each subsystem acts exclusively within its node, we maintain the product state structure essential for Theorem~\ref{thm:divide+conquer}. Therefore, our divide-and-conquer strategy can be adapted to fermionic systems, provided that the subdivision of subsystems is executed based on the set of orbitals rather than on the physical systems themselves.

\section{Conclusion}
\label{sec:conclusion}

Finding accurate approximations to the ground state of quantum systems is crucial for quantum computing applications in quantum chemistry and beyond~\cite{reiher2017elucidating,von2021quantum, lee2021even, Lee2023, Mitarai2022, Simon2024liouvillian}. Examples include drug design~\cite{caesura2025faster, Santagati2024, Cortes2024_SAPT} and battery optimization~\cite{kim2022fault}. 
Most quantum algorithms for ground state preparation require an easily preparable initial state that has a large overlap with the target ground state in order to run efficiently~\cite{Ge2019, Lin2020, Tubman2018}.
However, the Van Vleck catastrophe suggests that a suitable initial state, which has large overlap with the target ground state state, turns into a bad initial state exponentially fast in the sense that the overlap with the target ground state decreases exponentially with the system size, even for systems composed of separable subsystems~\cite{Vanvleck1936, Kohn1999}.

In this work, we propose a divide-and-conquer strategy for ground state preparation on quantum computers that can circumvent the exponential cost implied by the Van Vleck catastrophe. The main idea is to divide the system of interest into smaller subsystems and then use phase estimation repeatedly to project onto the ground state of larger and larger sets of subsystems until the whole system is rebuilt.
We show that, under certain assumptions on the Hamiltonian, the number of gates required for preparing the ground state of $N$ (weakly) interacting systems scales only like $O \lb N^{\log \log \lb N \rb} \mathrm{poly} (N) \rb$, which is quasi-polynomial in the number of systems (see Theorem~\ref{thm:divide+conquer} and Corollary~\ref{cor:no_van_vleck} for details). It might be possible to achieve a polynomial upper bound by conducting a tighter or more specialized analysis of our approach.

A key assumption in our divide-and-conquer strategy is that when two interacting subsystems are combined, the overlap between the product state of the two non-interacting ground states, $\ket{\psi_{s0}}$ and $\ket{\psi_{s1}}$, and the interacting ground state, $\ket{\psi_{s}}$, remains above a positive threshold $r$, i.e.  $\left| \bra{\psi_{s}}(\ket{\psi_{s0}}\ket{\psi_{s1}}) \right| \geq r$. To explore how the parameter $r$ relates to the properties of the Hamiltonian, we use perturbation theory to derive a bound on the achievable overlap (see Theorem~\ref{thm:perturbed_overlap}). This bound implies that our method works well when the spectral norm of the interaction Hamiltonian between the systems that are being combined is considerably smaller than the eigenvalue gap of the individual subsystem Hamiltonians.

The success of our divide-and-conquer protocol can also be related to the entanglement structure of the ground states that we are trying to construct.  In particular, we show that if the system exhibits a volume law of entanglement with respect to the subsystems being merged via the divide-and-conquer algorithm, then the divide-and-conquer algorithm will fail to produce the target ground state.  In essence, this tells us that if the ground state is strongly correlated across the partitions that we aim to glue together, then the entire system should be treated as a single large system or molecule. Trying to construct it from smaller components in this case is likely folly.  In contrast, if we have an area law scaling then the success probability can be large, but further assumptions (similar to those in Theorem~\ref{thm:divide+conquer}) need to be imposed to ensure that the success probability is indeed large.

We further complement our results with a numerical analysis of a one-dimensional transverse-field Ising model, for which our divide-and-conquer approach maintains a finite overlap at each recursive step, while the overlap of the initial state with the overall ground state vanishes exponentially with the system size. We find that although the conditions on the relationship between the spectral norm of the interaction Hamiltonian and the non-interacting spectral gap in Theorem~\ref{thm:perturbed_overlap} are not satisfied, the overlap at each recursive step does not necessarily vanish. This suggests the possibility of less restrictive sufficient conditions for when the divide-and-conquer approach will vanquish the Van Vleck catastrophe.

While our work illustrates that orthogonality catastrophes, such as the Van Vleck catastrophe, are not necessarily as damaging as they may seem, it does not imply that the divide-and-conquer strategy (or any other strategy) can fully address the ground state preparation problem.  This is a direct consequence of the fact that finding the ground state of $2$-local Hamiltonians is \QMA-hard~\cite{kempe2006complexity} and hence is impossible to do efficiently on a quantum computer unless $\BQP = \QMA$.  This is a serious concern facing the viability of probing ground state properties of physical systems.  As a community we need better techniques for understanding the cost of preparing approximate ground states for chemical systems of practical interest, especially for systems with strong correlations~\cite{reiher2017elucidating,Lee2023,goings2022reliably}. One task that needs to be addressed to achieve this involves identifying whether physically realistic molecules have ground states that are computationally difficult to prepare. Larger scale numerical studies may be needed to provide insight into this from a chemical perspective. More broadly, however, there is likely a need in the community as a whole to focus less on problems involving ground states of, e.g., molecules and focus more on excited states or dynamics~\cite{childs2014bose} as these problems are much more likely to be ones where quantum algorithms have a genuine advantage over classical algorithms.

\begin{acknowledgments}
    SS and NW were supported by the DOE, Office of Science, National Quantum Information Science Research Centers, Co-design Center for Quantum Advantage (C2QA) under Contract No. DE SC0012704 (Basic Energy Sciences, PNNL FWP 76274) and acknowledge funding from Boehringer Ingelheim. SS further acknowledges support from an Ontario Graduate Scholarship. We thank Rob Parrish and Clemens Utschig-Utschig for helpful discussions.
\end{acknowledgments}

\bibliography{refs}

\end{document}